\documentclass[twoside,letter]{revtex4} 
\usepackage{amsmath,graphicx,amssymb,fancyhdr,amsthm} 
\newtheorem{thm}{Theorem}[section] 
\newtheorem{cor}[thm]{Corollary}

\newtheorem{con}[thm]{Conjecture} 
\theoremstyle{definition} 
 
\theoremstyle{remark} 
 
\def\beq{\begin{eqnarray}} 
\def\eeq{\end{eqnarray}} 
\def\bsp{\begin{split}} 
\def\esp{\end{split}} 
\def\lra{\longrightarrow} 
\def\ra{\rightarrow} 
\def\inl{\left\langle} 
\def\inr{\right\rangle} 
\def\dt{\partial_t} 
\def\Tr{\mathrm{Tr}} 
\def\d{\mathrm{d}} 
\def\diag{\mathrm{diag}} 
\def\i{\mathrm{i}} 
\def\RC{$CSI_R$} 
\def\KC{$CSI_K$} 
\def\FC{$CSI_F$} 
\newcommand{\FCSI}[1]{$CSI_{F,#1}$} 
\newcommand{\mf}[1]{{\mathfrak #1}} 
\newcommand{\ul}[1]{\underline{#1}} 
\newcommand{\mb}[1]{{\mathbb #1}} 
\newcommand{\mc}[1]{{\cal #1}} 
\newcommand{\mbold}[1]{\mbox{\boldmath{\ensuremath{#1}}}}

\begin{document} 
 
\title{\Large\textbf{Lorentzian manifolds and scalar curvature invariants}} 
\author{{\large\textbf{Alan Coley$^{1}$, Sigbj\o rn Hervik$^{2}$, 
Nicos Pelavas$^{1}$} }  \vspace{0.3cm} \\ 
$^{1}$Department of Mathematics and Statistics,\\ 
Dalhousie University, 
Halifax, Nova Scotia,\\ 
Canada B3H 3J5 
\vspace{0.2cm}\\ 
$^{2}$Faculty of Science and Technology,\\ 
 University of Stavanger,\\  N-4036 Stavanger, Norway 
%\email{ 
\vspace{0.3cm} \\ 
\texttt{aac, pelavas@mathstat.dal.ca, sigbjorn.hervik@uis.no} } 
\date{\today} 
%\maketitle 
\pagestyle{fancy} 
\fancyhead{} % clear all header fields 
\fancyhead[EC]{A. Coley, S. Hervik and N. Pelavas} 
\fancyhead[EL,OR]{\thepage} 
\fancyhead[OC]{Lorentzian manifolds and scalar curvature invariants} 
\fancyfoot{} % clear all footer fields 
 
\begin{abstract} 
We discuss (arbitrary-dimensional) Lorentzian manifolds and the scalar polynomial curvature 
invariants constructed from the Riemann tensor and its covariant 
derivatives. Recently, we have shown that in four dimensions a Lorentzian 
spacetime metric is either $\mathcal{I}$-non-degenerate, and hence 
locally characterized by its scalar polynomial curvature 
invariants, or is a  
degenerate Kundt spacetime. We present a number of results that
generalize these  results to higher 
dimensions and discuss their consequences and potential physical 
applications. 
 
\end{abstract} 
\maketitle 
 
%%%%%%%%%%%%%%%%%%%%%%%%%%%%%%%%%%%%%%%%%%%%% 

\section{Introduction}

We address the question of when a Lorentzian manifold (in 
arbitrary dimensions) can be uniquely characterized (locally) by 
its scalar polynomial curvature invariants, which are scalars 
obtained by contraction from a polynomial in the Riemann tensor 
and its covariant derivatives. This question is not only of 
mathematical interest, but is also of fundamental physical import. 
We begin by introducing some necessary mathematical terminology 
and machinery.

For a spacetime $(\mathcal{M},g)$, a (one-parameter) \emph{metric 
deformation}, $\hat{g}_\tau$, $\tau\in [0,\epsilon)$, is a family 
of smooth metrics on $\mathcal{M}$ such that $\hat{g}_\tau$ is 
continuous in $\tau$,   $\hat{g}_0 = g$,   and $\hat{g}_\tau$ for 
$\tau>0$, is not diffeomorphic to $g$. We define the set of all 
scalar polynomial curvature invariants on $(\mathcal{M},g)$ by 
 
\begin{equation} 
\mathcal{I}\equiv\{R,R_{\mu\nu}R^{\mu\nu},R_{\mu\nu}R^{\mu}_{~\pi} R^{\pi\nu},\dots,
C_{\mu\nu\alpha\beta}C^{\mu\nu\alpha\beta}, \dots,
R_{\mu\nu\alpha\beta;\gamma}R^{\mu\nu\alpha\beta;\gamma}, \dots,
R_{\mu\nu\alpha\beta;\gamma\delta}R^{\mu\nu\alpha\beta;\gamma\delta},\dots\} \,. \nonumber 
\end{equation} 
If there does not exist a metric deformation of $g$ having the 
same set of invariants as $g$, then we will call the set of 
invariants \emph{non-degenerate}, and the spacetime metric $g$ 
will be called \emph{$\mathcal{I}$-non-degenerate} \cite{inv}. 
Therefore, for a metric which is $\mathcal{I}$-non-degenerate the 
invariants locally characterize the spacetime completely.

The Kundt class of spacetimes in $n$ dimensions is 
defined by those metrics admitting a null vector that is geodesic, 
expansion-free, shear-free and twist-free. A Kundt metric can be 
written in the canonical form ($i=3,...n$) \cite{coley,CSI} 
\beq \d 
s^2=2\d u\left[\d v +H(v,u,x^k)\d u+W_{i}(v,u,x^k)\d 
x^i\right]+h_{ij}(u,x^k)\d x^i\d x^j.\label{Kundt} \eeq 
A {\it degenerate Kundt} spacetime is a spacetime in which there 
exists a kinematic null frame (in which the appropriate Ricci 
rotation coefficients $L_{ij}$ are zero \cite{coley}) such that 
all of the positive boost weight (b.w.) terms of the Riemann tensor and 
all of its covariant derivatives $\nabla^{(k)} (Riem)$ are zero 
(in this common frame) \cite{Kundt}. That is, a degenerate Kundt 
spacetime is an {\it aligned algebraically special Riemann type 
$II$} and  {\it aligned algebraically special $\nabla^{(k)} 
(Riem)$ type $II$} Kundt spacetime. In terms of the metric 
(\ref{Kundt}), written in the canonical (kinematic) frame,  the 
condition that the  Riemann tensor is aligned and of algebraically 
special type $II$ implies that $W_{i,vv}=0$, and the condition 
that $\nabla(Riem)$ is aligned and of algebraically special type 
$II$ implies that $H_{,vvv}=0$ (whence it follows that all of 
$\nabla^{(k)} (Riem)$ ($k>1$) are aligned and of algebraically 
special type $II$) \cite{Kundt}. We note that the important 
constant curvature invariant ($CSI$) spacetimes \cite{CSI} and 
vanishing scalar invariant ($VSI$) spacetimes \cite{Higher} are 
degenerate Kundt spacetimes. 
 
Let us briefly review the results of \cite{inv}, particularly 
those that have generalizations to higher dimensions, in which the 
class of four-dimensional (4D) Lorentzian manifolds that can be 
completely characterized by the scalar polynomial curvature 
invariants constructed from the Riemann tensor and its covariant 
derivatives was determined. The important result that {\it a 
spacetime metric is either $\mathcal{I}$-non-degenerate or the 
metric is a degenerate Kundt metric} was proven. This theorem was 
proven on a case-by-case basis, depending on the algebraic type, 
using a b.w. decomposition and, most importantly, by 
determining an appropriate set of projection operators from the 
Riemann tensor and its covariant derivatives. 
We recall that the 4D Lorentzian manifolds are characterized 
algebraically by their Petrov and Segre \cite{kramer} types or, 
alternatively, in terms of their Ricci, Weyl (and Riemann) types 
\cite{coley,class}. 
 
It is useful to state a number of partial results of when a 
spacetime metric is $\mathcal{I}$-non-degenerate (which is how the 
theorem in 4D was actually proven), which will be exploited in 
obtaining and stating results in higher dimensions. First, it was 
proven that {\it if a 4D spacetime metric is locally of Ricci type 
$I$ or Weyl type $I$ (i.e., algebraically general) the metric is 
$\mathcal{I}$-non-degenerate} \cite{inv}. This indicates that, in 
general, a spacetime metric is $\mathcal{I}$-non-degenerate and the 
metric is locally determined by its curvature invariants. For the 
algebraically special cases the Riemann tensor itself does not 
provide sufficient information to determine all of the required 
projection operators, and it is necessary to also consider the 
covariant derivatives. In terms of the b.w. decomposition, 
for an algebraically special metric (which has a Riemann tensor 
with zero positive b.w. components) which is not Kundt, by 
taking covariant derivatives of the Riemann tensor positive boost 
weight components are found and a set of higher derivative 
projection operators are determined. Therefore, {\it if the 4D 
spacetime metric is algebraically special, but $\nabla R$, 
$\nabla^{(2)} R$, $\nabla^{(3)} R$, or $\nabla^{(4)} R$ is of type 
$I$ or more general, the metric is $\mathcal{I}$-non-degenerate.}

The remaining metrics which \emph{do not} acquire positive boost 
weight components when taking covariant derivatives have a very 
special curvature  structure; indeed, they are degenerate Kundt 
metrics \cite{Kundt}. This implies that metrics not determined by 
their curvature invariants must be of degenerate Kundt form; i.e., 
{\it degenerate Kundt metrics are not 
$\mathcal{I}$-non-degenerate}. This exceptional property of the 
the degenerate Kundt metrics essentially follows from the fact 
that they do not define a unique timelike curvature operator. 
The degenerate Kundt spacetimes are classified algebraically by 
the Riemann tensor and its covariant derivatives in the aligned 
kinematic frame \cite{Kundt}.

We note parenthetically that these results are of importance to 
the equivalence problem of characterizing Lorentzian spacetimes 
(in terms of their Cartan invariants) \cite{kramer}. Clearly, by 
knowing which spacetimes can be characterized by their scalar 
curvature invariants alone, the computations of the invariants 
(i.e., simple scalar invariants) is much more straightforward and 
can be done algorithmically. On the other hand, the Cartan 
equivalence method also contains, at least in principle, the 
conditions under which the classification is complete.

\section{Higher dimensions} 
 
Let us now consider higher dimensions. The results are proven on a 
case-by-case basis in terms of algebraic types; hence we need to 
utilize the higher dimensional algebraic classification of tensors 
\cite{class,coley}. 
We note that recently a number of exact higher  dimensional
algebraically special spacetimes have been studied \cite{POD}.
In particular, similar results to those that 
occur in 4D (discussed above) occur in the {\it algebraically 
general} cases. For Ricci type I  or G the proof is essentially 
identical to the 4D case. We know the Segre types in higher 
dimensions and for Ricci type I they are of the form 
$\{1,111\dots\}$, $\{1,(11)1\dots\}$, $\{1,(111)\dots\}$ etc., or 
$\{z\bar{z}11\dots\}$, $\{z\bar{z}(11)\dots\}$, etc.  We note that 
the stabilizer of a Ricci type I tensor is always contained in the 
compact group $O(n-1)$.  For all of these Segre types it is always 
possible to construct a time-like projector.  Therefore, by 
projecting the various curvature tensors we can construct tensors 
with purely spatial indices and, consequently, we can construct 
curvature invariants invariant under the compact group $O(n-1)$.
Therefore, in arbitrary dimensions, we have proven that:
\begin{thm} 
If a spacetime metric is 
of Ricci type I or G, the metric is $\mathcal{I}$-non-degenerate. 
\end{thm}

We can provide a compelling argument that in higher dimensions a 
general Weyl type I or G spacetime is 
$\mathcal{I}$-non-degenerate. A bivector formalism in higher 
dimensions is needed to characterize (invariantly) the curvature 
operators of the Weyl tensor \cite{Bivector}. In general, the Weyl 
tensor can be decomposed using the eigenspace projection operators 
as: 
\[ {\sf C}={\sf N}+\sum_A\lambda_A \bot_A,\] where ${\sf N}$ is 
nilpotent. In the case of the Weyl tensor, each of the projection 
operators, $\bot_A$, are of type D with respect to a certain 
frame. Now, if all of the projection operators are of type D with 
respect to the same frame, then the Weyl tensor is of type II or 
simpler. Therefore, for the Weyl tensor to be of type I/G, these 
frames cannot be aligned.  Each of these Weyl projection operators 
can be used to construct a projection operator of the tangent 
space of some type $\{(1,1..1)(1..1)\}$. Since these are not 
aligned, and the fact that Weyl tensor is of type I/G (and not 
simpler), by successive projections we can isolate a timelike 
direction and construct a timelike projection operator. It then 
follows that  {\it a Weyl type I or G spacetime is 
$\mathcal{I}$-non-degenerate}. Therefore, we have that:

\begin{con} 
If the Weyl type is I or G, then the spacetime is 
$\mathcal{I}$-non-degenerate. 
\end{con}

The first step in providing a rigorous proof to this result
is to investigate the curvature operators in higher dimensions and 
to classify these for the various algebraic types \cite{Bivector}. Indeed, with the aid of 
these operators is it then hoped that all of the results obtained in 4D 
can also be shown to be true in higher dimensions. In addition, such a formalism 
may lead to simpler proofs of some of the results outlined below. 

Finally, the result that a spacetime is 
$\mathcal{I}$-non-degenerate follows for any curvature operator 
(not just those constructed from the Ricci or Weyl tensors) that 
is of general algebraic type I or G:

\begin{cor}
If any curvature operator is of general algebraic type I or G, then the spacetime is 
$\mathcal{I}$-non-degenerate.
\end{cor}

We have thus far presented some general results; that is, in general 
(defined algebraically) a spacetime is 
$\mathcal{I}$-non-degenerate. It is also possible to present a 
very special result; namely, that a spacetime which is degenerate 
Kundt is not $\mathcal{I}$-non-degenerate.  By construction, all 
degenerate Kundt spacetimes with the same boost 
weight zero terms but with different negative b.w. terms will 
have precisely the same set of scalar curvature invariants since 
no negative boost-weight terms can appear in any scalar polynomial 
invariant in a degenerate Kundt spacetime (see theorem II.7 below). 
Therefore, we have that:
 
\begin{thm} 
A spacetime which is degenerate 
Kundt is not $\mathcal{I}$-non-degenerate. 
\end{thm}
 
In higher dimensions the intermediate cases are much harder to 
deal with (than in the 4D case). Let us first present a partial result. In the analysis in 4D 
it was determined for which Segre types for the Ricci tensor  the 
spacetime is $\mathcal{I}$-non-degenerate (similar results were 
obtained for the Weyl tensor). In each case, it was found that the 
Ricci tensor, considered as a curvature operator, admits a 
timelike eigendirection. Therefore, if a spacetime is not 
$\mathcal{I}$-non-degenerate, its Ricci tensor must be of a 
particular Segre type (corresponding to the non-existence of a 
unique timelike direction). Therefore, consider the following 
algebraic types for the Ricci tensor (or any other $(0,2)$ 
curvature operator written in `Segre form'):

\begin{enumerate} 
 
\item{} 
$\{2111...\}$,$\{2(11)1...\}$,$\{2(111)...\}$,...,$\{2(111...)\}$, 
\item{} 
$\{(21)11...\}$,$\{(21)(11)1...\}$,$\{(21)(111)...\}$,...,$\{(21)(111...)\}$, 
\item{} 
$\{(211)11...\}$,$\{(211)(11)1...\}$,$\{(211)(111)...\}$,...,$\{(211)(111...)\}$, 
\item{} 
$\{3111...\}$,$\{3(11)1...\}$,$\{3(111)...\}$,...,$\{3(111...)\}$, 
... 
\end{enumerate} 
It follows that if the Ricci tensor (or any $(0,2)$ curvature 
operator -- and similar results are true in terms of the Weyl 
tensor in bivector form) is not of one of these types (for 
example, of type $\{1,111...\}$), then the spacetime is 
$\mathcal{I}$-non-degenerate. It is plausible that if the Ricci 
tensor and all other curvature operators are {\it all} of one of 
these types, then the spacetime is degenerate Kundt and not 
$\mathcal{I}$-non-degenerate.

In addition, in higher dimensions suppose there exists a frame in 
which all of the positive b.w. terms of the Riemann tensor 
and all of its covariant derivatives $\nabla^{(k)} 
(Riem)$ are zero (in this frame) (i.e., the spacetime is of `type II to 
all orders'), it is plausible that the resulting spacetime is 
degenerate Kundt (i.e., the appropriate Ricci rotation coefficients 
$L_{ij}$ are zero \cite{coley}). This is true in 4D (as a result 
of the theorems of \cite{inv}). It is likely true in higher 
dimensions, but this might be difficult to prove in general.
In particular, 
we would like to prove that the 
degenerate Kundt metrics are the only metrics not determined by 
their curvature invariants (i.e., not 
$\mathcal{I}$-non-degenerate) in any dimension. It is hoped that
higher dimensional generalizations of all of the 4D Theorems can be 
proven with the aid of 
Weyl operators in higher dimensions \cite{Bivector}. However, we do note that 
many of the results in  \cite{inv,Kundt,CSINEW} {\em can be} generalized to higher 
dimensions. 

Let us present a number of partial results. 
First, two higher 
dimensional results were proven in  \cite{Kundt}.
Let $K_{n}$ denote the subclass of Kundt metrics such that there
exists a kinematic frame in which the Riemann tensor up to and including its
$n^{th}$ covariant derivative have vanishing positive  b.w. components.

\begin{thm} 
In the higher-dimensional Kundt class, $K_1$ implies $K_n$ for all 
$n\geq 2$. 
\end{thm}

\begin{thm} 
If for a spacetime, $(\mathcal{M},{\bf g})$, the Riemann tensor and 
all of its covariant derivatives $\nabla^{(k)}(Riem)$ are 
simultaneously of type D (in the same frame), then the spacetime 
is \emph{degenerate Kundt}. 
\end{thm} 
 Note that the degenerate Kundt spacetimes can be written in the form (\ref{Kundt}), where 
\[ H(v,u,x^k)=v^2H^{(2)}(u,x^k)+vH^{(1)}(u,x^k)+H^{(0)}(u,x^k), \quad  W_i(v,u,x^k)=vW^{(1)}_i(u,x^k)+W^{(0)}_i(u,x^k).\]

Let us next present a new result:
\begin{thm} 
For a degenerate Kundt spacetime the boost weight 0 components of all 
curvature tensors are identical to the corresponding Kundt 
spacetime where $H^{(1)}(u,x^k)=H^{(0)}(u,x^k)=W^{(0)}(u,x^k)=0$. 
Consequently, their curvature invariants will also be identical. 
\end{thm} 
\begin{proof} 
Let us introduce the null-frame (Kundt frame): 
\beq 
&&\ell=\d u, \quad {\bf n}=\d v+(H^{(2)} v^2+H^{(1)}v+H^{(0)})\d u+(W_i^{(1)}v+W^{(0)})\d x^i, \nonumber \\ 
&& {\bf m}^i={\sf e}^i_{~j}(u,x^k)\d x^j, \quad \delta_{ij}{\bf m}^i{\bf m}^j=g_{ij}(u,x^k)\d x^i\d x^j 
\eeq 
where the functions $H^{(2)}$, $H^{(1)}$, $H^{(0)}$, $W_i^{(1)}$ and $W^{(0)}$ do not depend on $v$.  
By calculating the Riemann tensor with respect to this null-frame, 
we get that the Riemann tensor, $R$, has the following b.w. decomposition:  
\beq 
R=(R)_0+(R)_{-1}+(R)_{-2}, 
\eeq 
where the components of $(R)_0$ do not depend on $v$ and involve
only the functions $H^{(2)}$ and $W_i^{(1)}$ and the transverse metric $g_{ij}(u,x^k)$.   
 
Consider now the $n$th-covariant derivatives, symbolically written $\nabla^{(n)}R$. 
By using the Kundt frame, a covariant derivative of an arbitrary covariant tensor $T$ 
can be written symbolically:  
\beq\label{cod} 
\nabla T=\partial T-\sum \Gamma*T, 
\eeq 
where $\partial$ are partial derivatives with respect to the frame, and
$\Gamma$ are the connection coefficients.  In the Kundt frame the
connection coefficents of positive b.w.  are all zero; consequently, the
piece $\sum \Gamma*T$ cannot raise the b.w.  Regarding the partial
derivatives, $\ell^\mu\partial_{\mu}\equiv \partial_v$ and ${\bf n}$
raises and lowers the b.w., respectively.  The partial derivatives
with respect to  ${\bf m}_i$ are of b.w.  0.  Therefore, if $T$ is of boost-order
0, then the b.w.  $+1$ and $0$ components of $\nabla T$ will be:
\begin{itemize}
\item{} b.w. $+1$: $\partial_v (T)_0$ 
\item{} b.w. $\phantom{+}0$: $(\Gamma)_0(T)_0$, ${\bf m}_i(T)_0$, $\partial_v(T)_{-1}$ 
\end{itemize} 
First, from theorem II.5, we note that since these are $K_1$ they are $K_n$, implying 
$(\nabla^{(n)}R)_{b>0}=0$; i.e., all positive b.w. components vanish. The highest b.w. terms
are thus the b.w. 0 components. We thus need to consider the b.w. 0 components in more detail.  
 
Consider first $\nabla R$. We note that the terms $(\Gamma)_0(R)_0$, ${\bf m}_i(R)_0$ all 
give the desired result. However, we need to investigate $\partial_v(R)_{-1}$ in more 
detail. This comes from the covariant derivative 
$\ell^\mu\nabla_{\mu}(R)_{-1}\equiv \nabla_{+}(R)_{-1}$. Now, we observe that the Bianchi identity $R_{\mu\nu(\alpha\beta;\delta)}=0$ enables us to rewrite the troublesome derivatives  $ \nabla_{+}(R)_{-1}$ in terms of derivatives with respect to ${\bf n}$ and ${\bf m}^i$: 
\[ R_{-ijk;+}=-R_{-ik+;j}-R_{-i+j;k}, \quad R_{-+-i;+}=-R_{-+i+;-}-R_{-++-;i}.\] 
Consequently, the components of  $ \nabla_{+}(R)_{-1}$ can be written in terms of other 
well-behaving b.w. 0 terms.  
 
Let us now show that we can always write the b.w. 0 terms
$\nabla_+(\nabla^{(n)}R)_{-1}$, in terms of well-behaving b.w. 0
components of lower or equal order.  We will show this by
induction; therefore, assume  it is true for
$\nabla_+(\nabla^{(n)}R)_{-1}$, $n\geq 1$.  These components have the
form $R_{\mu\nu\alpha\beta;\delta_1...\delta_n}$ where, by the induction
assumption, $\delta_n=\{-,i\}$.  We need to check the components of
$\nabla_+(\nabla^{(n)}R)_{-1}$:
$R_{\mu\nu\alpha\beta;\delta_1...\delta_n +}$.

By the generalised Ricci identity we have that: 
\[ R_{\mu\nu\alpha\beta;\delta_1...\delta_{n-1}\delta_n +}=R_{\mu\nu\alpha\beta;\delta_1...\delta_{n-1}+\delta_n}+\sum \left[R*\nabla^{(n-1)}R\right]_{\mu\nu\alpha\beta\delta_1...\delta_{n-1}\delta_n+}\] 
We notice that all of the terms on the right-hand side are well-behaved; therefore, 
the left-hand side is well-behaved also. Morever, we also see that all of the b.w. 0 terms of 
the form $\nabla_+(\nabla^{(n)}R)_{-1}$ can be written in terms of well-behaved components of 
lower or equal order. Therefore, by induction,  components of $(\nabla^{(n)}R)_0$ are 
well-behaved for all $n\geq 0$. The theorem is consequently proven.  
\end{proof}

Note that an alternative proof might be given by using
$\nabla=\widetilde{\nabla}+\tau$, where 
$\widetilde{\nabla}$ is the corresponding connection with  
$H^{(1)}(u,x^k)=H^{(0)}(u,x^k)=W^{(0)}(u,x^k)=0$ and $\tau$ is a tensor 
(the remaining piece). The tensor $\tau$ will be of boost order $-1$ (and so on).

Furthermore, there are a number of results that are a direct
consequence of this theorem. For example, we can give 
conditions for when a Kundt spacetime is either 
$\mathcal{I}$-symmetric or Kundt-$CSI$. We recall that for $CSI$ metrics, 
$\mathcal{I}$-non-degeneracy implies that the spacetime is 
curvature homogeneous to all orders; hence, an important corollary 
of the results of \cite{inv} is a proof of the $CSI$-Kundt 
conjecture in 4D \cite{CSINEW}, that for {\it  a 4D $CSI$ spacetime 
then either the spacetime is locally homogeneous or a subclass of 
the Kundt spacetimes}. It is plausible that this result 
generalizes to higher dimensions. In the context of string theory, it is of considerable interest to 
study higher dimensional Lorentzian $CSI$ spacetimes. In 
particular, a number of N-dimensional $CSI$ spacetimes are known 
to be solutions of supergravity theory when supported by 
appropriate bosonic fields \cite{CFH}.

First, we have that:

\begin{cor} 
A degenerate Kundt spacetime is $\mathcal{I}$-symmetric if and only if the corresponding spacetime with $H^{(1)}(u,x^k)=H^{(0)}(u,x^k)=W^{(0)}(u,x^k)=0$ is also $\mathcal{I}$-symmetric. 
\end{cor} 

We can also prove the following:

\begin{thm} 
If the spacetime is of type D$^k$, then the components of the 
curvature tensors are determined by the scalar curvature invariants. 
\end{thm} 
\begin{proof} 
First we note that a type D$^k$ spacetime possesses at least one Killing vector (actually, at least three, see Corollary that follows), namely a boost isotropy. 
Therefore, let  $H= SO(1,d)$ be the isotropy group of all the curvature tensors 
(which necessarily must be at least of dimension 1; i.e., $d\geq 1$). 
Let $g_{AB}$ be the projector onto the tangent subspace of the orbits of $H$. 
Furthermore, let $g_{IJ}$ be defined so that $(g_{\mu\nu})=(g_{AB})\oplus(g_{IJ})$. 
This implies that any curvature tensor can be orthogonally decomposed as a tensor; thus, 
symbolically
\[ \nabla^{(k)}R= \sum \phi{\otimes}\mathcal{R}, \]  
where $\phi=(\phi_{ABC...D})$ is a scalar representation of $H$, and 
$\mathcal{R}=(\mathcal{R}_{IJ...K})$ is a tensor over $g_{IJ}$. Consequently, 
if $1+d+\tilde{d}=$ dimension of spacetime, then  $\nabla^{(k)}R$ can be considered as an 
$SO(\tilde{d})$-tensor. Hence, all scalar curvature invariants can be considered as 
$SO(\tilde{d})$-invariants. Since this group is compact, the group action separates 
orbits, and thus the components are determined by the scalar invariants.   
\end{proof} 
Note that this means that we can also say that a 
type D$^k$ spacetime is characterized by its invariants 
(however, in a different sense than $\mathcal{I}$-non-degeneracy) \cite{Bivector}.
This result also immediately implies that, for example, a spacetime which is of type D$^k$ 
and $CSI$ is necessarily Kundt and homogeneous \cite{Kundt}. In fact, in general, type D$^k$ spacetimes possess at least three Killing vectors: 
\begin{cor}
A spacetime of type D$^k$ possesses \emph{three} Killing vector fields with  
2-dimensional timelike orbits. 
\end{cor}
\begin{proof}
We have already pointed out that such a spacetime possesses a boost isotropy. Consider the Lie derivatives $\pounds_{\ell}I$, and $\pounds_{\bf n}I$ where $I$ is any polynomial curvature invariant. Since this is of type D$^k$, we must have: 
\[ \pounds_{\ell}I=\ell(I)=\ell^{\mu}\nabla_{\mu}I=0 \] 
(similarly for ${\bf n}$). Therefore, $\ell^{\mu}\nabla_{\mu}R=0$ for any curvature tensor $R$ and thus there exists a Killing vector $\tilde{\ell}$ (similarly for ${\bf n}$). Hence, the spacetime possesses three Killing vectors ($\tilde{\ell}$, $\tilde{\bf n}$, and a boost). 

\end{proof}

Perhaps a more useful consequence is:
\begin{cor} 
For a spacetime of (aligned) type II to all orders, the boost weight 0 components 
are determined by the curvature invariants.
\end{cor}
This corollary follows simply from the fact that when constructing a complete contraction of an arbitrary tensor of type II, only the b.w. 0 components will contribute; i.e., if $T$ is of type II, meaning $T=(T)_0+(T)_{-1}+(T)_{-2}+...$, then for a complete contraction: 
\[ \mathrm{Contr}[T]=\mathrm{Contr}[(T)_0].\] 
Consequently, the invariants of $T$ and $(T)_0$ are identical. Thus, since $(T)_0$ is of type D, the above theorem implies that its components are determined from the invariants.

This further implies a proof of the `$CSI$$_F$' conjecture:
\begin{cor}
A (degenerate) Kundt $CSI$  spacetime is a spacetime for which there exists a frame with a null vector
$\ell$ such that all components of the Riemann tensor and its
covariants derivatives in this frame have the property that (i)
all positive boost weight components (with respect to $\ell$) are zero and (ii) all
zero boost weight components are constant.
\end{cor} 

Finally, it should be possible  to prove an extension of the type D$^k$ 
result stated in \cite{Kundt}:

\begin{con} 
If the curvature tensors to all orders are type II (or simpler) and aligned 
(i.e., $\nabla^{(n)}R $ is of type II), then the spacetime is Kundt. 
\end{con} 

It then follows, under some appropriate assumptions, that
{\em degenerate Kundt spacetimes with $H^{(1)}=H^{(0)}=W_i^{(0)}=0$ {\bf } are 
of Riemann type D$^k$}.

In the future we  hope to establish all of the 
higher dimensional generalizations of the results obtained in 4D.  
As noted above, the first step 
is to investigate the curvature operators in higher dimensions and 
to classify these for the various algebraic types \cite{Bivector}.
 
\section{Discussion} 
 
We have found that the degenerate Kundt metrics are \emph{not 
determined by their curvature invariants} (in the sense that they 
are not $\mathcal{I}$-non-degenerate). Degenerate Kundt spacetimes 
are also special in a number of other ways including, for example, 
their holonomy structure, which may lead to novel and fundamental 
physics. Indeed, in a degenerate Kundt spacetime it is not 
possible to define a unique timelike curvature operator, and hence 
a unique timelike direction, and a 1+($n-1$) spacetime splitting of the 
spacetime is not possible.

Supersymmetric solutions of supergravity theories have played an 
important role in the development of string theory (see, for 
example, \cite{grover}) The existence of Killing spinors accounts 
for much of the interest in metrics with special holonomy in 
mathematical physics. Supersymmetric solutions in $M$-theory that 
are not static admit a {\it covariantly constant null vector} 
($CCNV$) \cite{coley}. The isotropy subgroup of a null spinor is 
contained in the isotropy subgroup of the null vector, which in 
arbitrary dimensions is isomorphic to the spin cover of 
$ISO(n-2)$.  A $CCNV$ metric is a degenerate Kundt metric 
(\ref{Kundt}) with $H_{,v} =0$ and $W_{i,v} = 0$. This class 
includes a subset of the Kundt-$CSI$ and the $VSI$ spacetimes as 
special cases. The $VSI$ and $CSI$ degenerate Kundt spacetimes are 
of fundamental importance since they are solutions of supergravity 
or superstring theory, when supported by appropriate bosonic 
fields \cite{CFH}.

The classification of holonomy groups in Lorentzian spacetimes is 
quite different from the Riemannian case. A Lorentzian manifold 
$\mathcal{M}$ is either {\it completely reducible}, and so 
$\mathcal{M}$ decomposes into irreducible or flat Riemannian 
manifolds and a manifold which is an irreducible or a flat 
Lorentzian manifold or $(\mathbb{R},-dt)$, or $\mathcal{M}$ is 
{\it not completely reducible}, which leads to the existence of a 
degenerate (one-dimensional) holonomy invariant lightlike subspace 
(the Lorentzian manifold decomposes into irreducible or flat 
Riemannian manifolds and a Lorentzian manifold with 
indecomposable, but non-irreducible holonomy representation), 
which gives rise to the {\it recurrent null vector} ($RNV$) and 
$CCNV$ (Kundt) spacetimes \cite{holon,johan}. [A $RNV$ metric has 
holonomy ${\rm Sim}(n-2)$ and the metric belongs to the class of 
Kundt metrics (\ref{Kundt}) with $W_{i} = W_{i}(u,x^k)$, but is 
not necessarily a degenerate Kundt metric.]

Therefore, the  Kundt spacetimes that are of particular physical 
interest are degenerately reducible, which leads to complicated 
holonomy structure and various degenerate mathematical properties. 
Indeed, it could be argued that a complete understanding of string 
theory is not possible without a comprehensive knowledge of the 
properties of the Kundt spacetimes \cite{johan}. For example, as 
noted above, a degenerate Kundt spacetime is not completely 
classified by its set of scalar polynomial curvature invariants 
(i.e., they have important geometrical information that is not 
contained in the scalar invariants). All {$VSI$} spacetimes and 
{$CSI$} spacetimes that are not locally homogeneous (including the 
important $CCNV$ subcase) belong to the degenerate Kundt class 
\cite{coley,class}. In these spacetimes all of the scalar 
invariants are constant or zero. This leads to interesting 
problems with any physical property that depends essentially on 
scalar invariants, and may lead to ambiguities and pathologies in 
models of quantum gravity or string theory.

As an illustration, in many theories of fundamental physics there 
are geometric classical corrections to general relativity. 
Different polynomial curvature invariants (constructed from the 
Riemann tensor and its covariant derivatives) are required to 
compute different loop-orders of renormalization of the 
Einstein-Hilbert action. In specific quantum models such as 
supergravity there are particular allowed local counterterms 
\cite{sugra}. All classical corrections are zero in $VSI$ spacetimes (and constant in 
$CSI$ spacetimes). Indeed, it is possible that a Lorentzian Kundt spacetime 
does not even allow for a low order perturbative expansion.

\section*{Acknowledgments} 
This work was supported by the 
Natural Sciences and Engineering Research Council of Canada.

\end{document}